\documentclass{article}
\usepackage[utf8]{inputenc}
\usepackage[T1]{fontenc}
\usepackage[english]{babel}
\usepackage{ntheorem}
\usepackage{amsmath,amssymb,amstext}
\usepackage{dsfont}
\usepackage{stmaryrd}

\newtheorem{definition}{Definition}

\newtheorem{proposition}{Proposition}
\newtheorem{proof}{Proof}

\newtheorem{theorem}{Theorem}

\newtheorem{remark}{Remark}
\newtheorem{lemma}{Lemma}

\title{Chaos in DNA inversions\\(Draft paper)}

\author{Jacques M. Bahi \and Nathalie C\^ot\'e \and Christophe Guyeux }

\begin{document}

\maketitle

\begin{abstract}
In this paper, reasons explaining why the CM model of Bahi and Michel (2008) simulates with a good accuracy genes mutations over time are proposed.
It is firstly justified that the CM model is a chaotic one, as it is defined by Devaney.
Then, it is established that inversions occurring in genes mutations have indeed a chaotic dynamic, thus making relevant the use of chaotic models for genes evolution. 
\end{abstract}

\textit{keywords:}
Genes evolution models; Inversions; Mathematical topology; Devaney's chaos.

\section{Introduction}

Codons are not uniformly distributed into the genome.
Over time mutations have introduced some variations in their apparition frequency.
It can be attractive to study the genetic patterns (blocs of more than one nucleotide: dinucleotides, trinucleotides...) that appear and disappear depending on mutation parameters.
Mathematical models allow the prediction of such an evolution, in such a way that statistical values observed into current genomes can be recovered.

A first model for genomes evolution has been proposed in 1969 by Thomas Jukes and Charles Cantor \cite{Jukes69}. 
This first attempt has been followed up by Motoo Kimura \cite{Kimura80}, Joseph Felsenstein \cite{Felsenstein1980}, Masami Hasegawa, Hirohisa Kishino, and Taka-Aki Yano \cite{Hasegawa1985} respectively.
The differences between these models are in the number of parameters they use, but all of these models manipulate constant parameters.
However, they are rudimentary as they only allow to study nucleotides evolution, not genetic patterns mutations.

From 1990 to 1994, Didier Arquès and Christian Michel have proposed models based on the RY purine/pyrimidine alphabet \cite{Arques1990,Arques1990a,Arques1992,Arques1993a,Arques1993b,Arques1994}. 
These models have been abandoned by their own authors in favor of models over the $\{A,C,G,T\}$ alphabet.
More precisely, Didier Arquès, Jean-Paul Fallot, and Christian Michel have proposed in \cite{Arques1998} a first evolutionary model on the $\{A,C,G,T\}$ alphabet that is based on trinucleotides.
As for the nucleotides based models, this new approach has taken into account only constants parameters.

In 2004, Jacques M. Bahi and Christian Michel have published a novel research work in which the model of 1998 has been improved by replacing constants parameters by new parameters dependent on time \cite{Bahi2004}.
By this way, it has been possible to simulate a genes evolution that is non-linear.
However, the following years, these researchers have been returned to models embedding constant parameters, probably due to the fact that the model of 2004 lead to poor results: only one of the twelve studied cases allows to recover values that are close to reality.
For instance, in 2006, Gabriel Frey and Christian Michel have proposed a model that uses 6 constant parameters \cite{Frey2006a}, whereas in 2007, Christian Michel has constructed a model with 9 constants parameters that generalize those of 1998 and 2006 \cite{Michel2007c}.
Finally, Jacques M. Bahi and Christian Michel have recently introduced in \cite{Bahi2008, Bahi2008a}, a last model with 3 constant parameters, but whose evolution matrix evolves over time.
In other words, trinucleotides that have to mutate are not fixed, but they are randomly picked among a subset of potentially mutable trinucleotides.

This model, called ``chaotic model'' CM, allow a good recovery of various statistical properties detected into the genome.
Furthermore, this model match well with the hypothesis of some primitive genes that have mutated over time.
In this paper, we wonder why the CM model gives good results. Obviously, to suppose that not all of the trinucleotides have to mutate at each time is reasonable as, for instance, the stop codons have very small mutation probabilities.
However, such a biological claim is not sufficient to explain all the consequences of the success of the CM model to simulate the dynamics of mutations into genomes. 
Indeed, we have recently established that such a model based on chaotic iterations is indeed really chaotic, as it is defined in the mathematical theory of chaos.
Before this proof, the term ``chaotic'' in these discrete iterations was only an adjective, having apparently no obvious relation with the well-established Devaney's characterization of an unpredictable behavior.

In this paper, we wonder why a model having a chaotic dynamics gives, in a certain way, better results than the standard model, to predict the evolution of genomes through mutations.
We will show that an important mutation mechanism, namely the inversion, has a chaotic dynamic over time. 
Consequences of this proof, for biology and evolution models, will finally be discussed.

The remainder of this research work is organized as follows.

\section{Discrete time chaotic evolution model}

\subsection{Chaotic Iterations}
\label{sec:chaotic iterations}

Let us consider  a \emph{system} with a finite  number $\mathsf{N} \in
\mathds{N}^*$ of elements  (or \emph{cells}), so that each  cell has a
Boolean  \emph{state}. A  sequence of  length $\mathsf{N}$  of Boolean
states of  the cells  corresponds to a  particular \emph{state  of the
system}. A sequence which  elements belong to $\llbracket 1;\mathsf{N}
\rrbracket $ is called a \emph{strategy}. The set of all strategies is
denoted by $\mathbb{S}.$

\begin{definition}
\label{Def:chaotic iterations}
The      set       $\mathds{B}$      denoting      $\{0,1\}$,      let
$f:\mathds{B}^{\mathsf{N}}\longrightarrow  \mathds{B}^{\mathsf{N}}$ be
a  function  and  $S\in  \mathbb{S}$  be  a  strategy.  The  so-called
\emph{chaotic      iterations}     are     defined      by     $x^0\in
\mathds{B}^{\mathsf{N}}$ and
$$     
\forall    n\in     \mathds{N}^{\ast     },    \forall     i\in
\llbracket1;\mathsf{N}\rrbracket ,x_i^n=\left\{
\begin{array}{ll}
  x_i^{n-1} &  \text{ if  }S^n\neq i \\  
  \left(f(x^{n-1})\right)_{S^n} & \text{ if }S^n=i.
\end{array}\right.
$$
\end{definition}

In other words, at the $n^{th}$ iteration, only the $S^{n}-$th cell is
\textquotedblleft  iterated\textquotedblright .  Note  that in  a more
general  formulation,  $S^n$  can   be  a  subset  of  components  and
$\left(f(x^{n-1})\right)_{S^{n}}$      can     be      replaced     by
$\left(f(x^{k})\right)_{S^{n}}$, where  $k<n$, describing for example,
delays  transmission~\cite{Robert1986}.

\subsection{Genes mutations shown as chaotic iterations}

When considering the model of 2007 with 9 constant parameters that generalize models of 1998 and 2006, all of the trinucleotides have to mutate at each time.
These models do not take into account the low mutability of the stop codons.
Additionally, they do not allow to apply mutation strategies on certain given codons, while the other codons do not mutate.
This is why a new model with 3 constant parameters has been proposed in \cite{Bahi2008,Bahi2008a}.
In this model, the set of trinucleotides is divided into two subsets at each time $t$: the first one is constituted by trinucleotides that can possibly mutate at time $t$, whereas in the second one trinucleotides cannot mutate at the considered time. 
The trinucleotides that mutate at time $t$ are randomly picked following a uniform distribution.
Consequently, the size and the constitution of the subset of mutable trinucleotides change at each time $t$. This subset is denoted by $J(t)$, and this new model has been called ``chaotic model'' by the authors of \cite{Bahi2008,Bahi2008a}, as opposed to the former ``standard model'' of 1998.

In the chaotic model, non-mutable trinucleotides cannot have been obtained by the mutation of other trinucleotides. 
Consequently, their probability of occurrence is constant, so their derivation is null.
Conversely, mutation parameters of the mutable trinucleotides are those of the model of 1998: $p$, $q$, and $r$ with $p+q+r=1$, for each of the three sites of nucleotides.

The new model is thus defined by the following way:
$$
\left\{
\begin{array}{ll}
P_i'(t) = 0  & \text{if } i \notin J(t)\\
P_i'(t) = \displaystyle{\sum_{j=1}^{64}} (A^{(t)}-I)_{ji}P_j(t) & \text{if } i \in J(t)
\end{array}
\right.
$$

Obviously, this new model is a generalization of the one of 1998, as if we suppose that, $\forall t, A^{(t)} = A$, and if $\forall t, J(t)$ is the set of all the trinucleotides, then the system above can be resumed to its second line, which is exactly the model of 1998.

As the number of mutable trinucleotides changes over time, the mutation matrix is not constant, which leads to the fact that the resolution method used in the standard model cannot be applied here.
To solve the system, authors of \cite{Bahi2008,Bahi2008a} have considered discrete times small enough to be sure that the mutation matrix do not change between two instants $t_i$ and $t_{i+1}$.

Let $A^{(k)}$ be the (constant) mutation matrix during the time interval $[t_{k-1}, t_k]$. 
To be able to compute $P_i'(t_{k-1})$, authors of \cite{Bahi2008,Bahi2008a} have used the Euler method, to obtain:

$$\dfrac{d~P_i(t_{k-1})}{dt} = \dfrac{P_i(t_k) - P_i(t_{k-1})}{h},$$

\noindent where $h=t_k-t_{k-1}$ is supposed small and constant. 
By putting this formula into the previous system, these authors have finally obtained:
$$
\left\{
\begin{array}{ll}
P_i(t_k) = P_i(t_{k-1})  & \text{if } i \notin J(t),\\
P_i(t_k) = h \displaystyle{\sum_{j=1}^{64}} (A^{(k)}-I)_{ji}P_j(t_{k-1}) + P_i(t_{k-1}) & \text{if } i \in J(t).
\end{array}
\right.
$$

This model has been called the ``discrete time chaotic evolution model CM'' in \cite{Bahi2008,Bahi2008a}.
This discrete version of the continuous chaotic one is, indeed, a gene evolution model that use chaotic iterations of Definition \ref{Def:chaotic iterations}.
To understand the interest of this discrete time chaotic evolution model, we must firstly recall the discovery by Michel \emph{et al.} of a $C^3-$code and its properties.

\subsection{Relevance of the CM model}

A computation of the frequency of each trinucleotide in the 3 frames of genes, in a large gene population (protein coding region) of both eukariotes and prokaryotes, has established in 1996 that the distribution of trinucleotides in these frames is not uniform.
Such a surprising result has led to the definition of 3 subsets of trinucleotides, denoted by $X_0$, $X_1$, and $X_2$.
$X_0$, $X_1$, and $X_2$ are respectively constituted by 20 trinucleotides. 
They are linked by the following  permutation property: $X_1 = \{\mathcal{P}(t), t \in X_0 \}$, $X_2 = \{\mathcal{P}(t), t \in X_1 \}$, where for all trinucleotide $t=n_0n_1n_2$, $\mathcal{P}(t)=n_2n_1n_0$.
More details about the research context and the properties of these sets ($C^3$ code, rarity, largest window length, higher frequency of ``misplaced'' trinucleotides, flexibility) can be found in \cite{Bahi2008,Bahi2008a}.
Among other things, it has been proven that $X_0$ occurs with the highest probability (48.8\%) in genes (reading frames 0), whereas $X_1$ and $X_2$ occur mainly in the frames 1 and 2, respectively. 
In other words, $X_0$ is not pure (its probability is less than 1): it is mixed with $X_1$ and $X_2$ in genes. 

Such a property can be explained as follows: random mutations have introduced noise during evolution, leading to a decreased probability of $X_0$ \cite{Bahi2008,Bahi2008a}.
Moreover, codes $X_1$ and $X_2$ are not symmetric in genes, i.e., $P(X_1) < P(X_2)$ (the probability difference is 4,8\%).
This is totally unexpected: the complementarity property should lead to the same probabilities for $X_1$ and $X_2$, even when considering noise during evolution.

The standard and chaotic models (with particular strategies for the stop codons) can explain both the decreased probability of the code $X_0$ and the asymmetry between the codes $X_1$ and $X_2$ in genes.
These standard and chaotic models construct ``primitive'' genes, i.e., genes before random substitutions, with trinucleotides of the circular code $X_0$. 
These models are able to find the frequency orders of the three codes $X_0$, $X_1$, and $X_2$ in genes. 
In particular, the chaotic model called ``$CM_{TAA}$'' with low mutability of the stop codon TAA, matches the probability discrepancy between the circular codes $X_1$ and $X_2$ observed in real genes.
Its ability to match is better than the standard model SM and the other chaotic models.

In the following section, we will propose some reasons explaining why some chaotic models match with a good accuracy frequency orders of the three codes, when considering the circular code $X_0$ as constituting the ``primitive'' genes.
More precisely, we will show that some genes evolution mechanisms are chaotic according to Devaney, thus explaining why chaotic models fit such evolution.

\section{The CM model is a truly chaotic one}

First of all, let us recall that the term  ``chaotic'', in  the name of  these iterations, has \emph{a priori} no link with the mathematical theory of chaos, recalled below.

\subsection{Devaney's chaotic dynamical systems}
\label{sec:Devaney}

Consider  a topological  space $(\mathcal{X},\tau)$  and  a continuous
function $f$ on $\mathcal{X}$.

\begin{definition}
  $f$ is said  to be \emph{topologically transitive} if,  for any pair
  of open sets $U,V \subset \mathcal{X}$, there exists $k>0$ such that
  $f^k(U) \cap V \neq \varnothing$.
\end{definition}

\begin{definition}
  An element  (a point) $x$  is a \emph{periodic element}  (point) for
  $f$  of period  $n\in \mathds{N}^*,$  if $f^{n}(x)=x$.
\end{definition}

\begin{definition}
  $f$ is said to be \emph{regular} on $(\mathcal{X}, \tau)$ if the set
  of periodic points for $f$  is dense in $\mathcal{X}$: for any point
  $x$ in $\mathcal{X}$, any neighborhood  of $x$ contains at least one
  periodic point.
\end{definition}

\begin{definition}
  $f$ is said  to be \emph{chaotic} on $(\mathcal{X},\tau)$  if $f$ is
  regular and topologically transitive.
\end{definition}

The   chaos   property  is   strongly   linked   to   the  notion   of
``sensitivity'', defined on a metric space $(\mathcal{X},d)$ by:

\begin{definition}
  \label{sensitivity} $f$ has \emph{sensitive dependence on initial conditions}
  if there  exists $\delta >0$  such that, for any  $x\in \mathcal{X}$
  and  any  neighborhood  $V$  of  $x$,  there  exists  $y\in  V$  and
  $n\geqslant 0$  such that $d\left(f^{n}(x),  f^{n}(y)\right) >\delta
  $.

  \noindent $\delta$ is called the \emph{constant of sensitivity} of $f$.
\end{definition}

Indeed, Banks  \emph{et al.}  have proven in~\cite{Banks92}  that when
$f$ is chaotic and $(\mathcal{X}, d)$  is a metric space, then $f$ has
the  property  of sensitive  dependence  on  initial conditions  (this
property was formerly  an element of the definition  of chaos). To sum
up, quoting Devaney in~\cite{Devaney}, a chaotic dynamical system ``is
unpredictable   because  of  the   sensitive  dependence   on  initial
conditions. It cannot be broken down or simplified into two subsystems
which do not interact because  of topological transitivity. And in the
midst  of this  random behavior,  we nevertheless  have an  element of
regularity''.  Fundamentally   different  behaviors  are  consequently
possible and occur in an unpredictable way.

\subsection{Chaotic iterations and Devaney's chaos}
\label{sec:topological}

In this section we give outline proofs of the properties establishing the fact that the CM model is truly chaotic, as it is defined in the Devaney's theory. 
The complete theoretical framework is detailed in~\cite{bg10:ij}.

Denote by $\Delta $ the \emph{discrete Boolean metric},
$\Delta(x,y)=0\Leftrightarrow x=y.$ Given a function $f$, define the
function: $F_{f}: \llbracket1;\mathsf{N}\rrbracket\times
\mathds{B}^{\mathsf{N}} \longrightarrow \mathds{B}^{\mathsf{N}}
$ such that $F_{f}(k,E)=\left( E_{j}.\Delta (k,j)+f(E)_{k}.\overline{\Delta
(k,j)}\right)_{j\in \llbracket1;\mathsf{N}\rrbracket}$.

Let us consider the phase space
$\mathcal{X}=\llbracket1;\mathsf{N}\rrbracket^{\mathds{N}}\times
\mathds{B}^{\mathsf{N}}$ and the map $G_{f}\left( S,E\right) =\left( \sigma (S),F_{f}(i(S),E)\right)
$, where $\sigma$ is defined by $\sigma :(S^{n})_{n\in \mathds{N}}\in \mathbb{S}\rightarrow (S^{n+1})_{n\in \mathds{N}}\in \mathbb{S}$,
and $i$ is the map $i:(S^{n})_{n\in \mathds{N}}\in \mathbb{S}\rightarrow S^{0}\in
\llbracket1;\mathsf{N}\rrbracket$. So the chaotic iterations can be described by the following iterations:
$$X^{0}\in \mathcal{X}\text{ and }X^{k+1}=G_{f}(X^{k}).$$

We have defined in~\cite{bg10:ij} a new distance $d$ between two points $(S,E),(\check{S},\check{E} )\in \mathcal{X}$
by
$d((S,E);(\check{S},\check{E}))=d_{e}(E,\check{E})+d_{s}(S,\check{S}),$
where:
\begin{itemize}
\item
$\displaystyle{d_{e}(E,\check{E})}=\displaystyle{\sum_{k=1}^{\mathsf{N}}\Delta
(E_{k},\check{E}_{k})} \in \llbracket 0 ; \mathsf{N} \rrbracket$,
\item
$\displaystyle{d_{s}(S,\check{S})}=\displaystyle{\dfrac{9}{\mathsf{N}}\sum_{k=1}^{\infty
}\dfrac{|S^{k}-\check{S}^{k}|}{10^{k}}} \in [0 ; 1].$
\end{itemize}

It is then proven that,

\begin{proposition}
\label{Prop:continuite} $G_f$ is a continuous function on $(\mathcal{X},d)$.
\end{proposition}

In the metric space $(\mathcal{X},d)$,
the vectorial negation $f_{0} :\  \mathbb{B}^N  \longrightarrow  \mathbb{B}^N $, $(b_1,\cdots,b_\mathsf{N})  \longmapsto (\overline{b_1},\cdots,\overline{b_\mathsf{N}})$ satisfies the three conditions for Devaney's
chaos: regularity, transitivity, and sensitivity~\cite{bg10:ij}. So,
\begin{proposition}
$G_{f_0}$ is a chaotic map on $(\mathcal{X},d)$ according to Devaney.
\end{proposition}

Thus the model that gives the best results, in a certain way, to the problem of genes evolution prediction, is a chaotic model.
We will give in the next sections a result concerning inversions that can possibly explain this fact, at least partially.

\section{How to Formalize Inversions}

\subsection{The inversion operator}

Let $\mathcal{N} = \{A,T,C,G\}$ be the set of nucleotides and $\mathsf{N} \in \mathds{N}^*$.
A chromosome with $\mathsf{N}$ nucleotides is any element of $\mathcal{N}^\mathsf{N}$. 
Let $\mathcal{C}$ be the set of all chromosomes of size $\mathcal{N}$.
For each $C,C' \in \mathcal{C}^\mathsf{N}$, the chromosome $\mathcal{C}$ is said to be changeable into $\mathcal{C}'$ if and only if there is a permutation mapping $\mathcal{C}$ into $\mathcal{C'}$. We denote it by $\mathcal{C} \approx \mathcal{C'}$.
In a mathematical point of view, $\approx$ is a relation of equivalency. In a biological point of view, the class of equivalency $\dot{\mathcal{C}}$ of $\mathcal{C}$ corresponds to all of the possible and conceivable reordering of the chromosome $\mathcal{C}$ over time. 
By reordering, we mean a simple change of the order of the nucleotides into $\mathcal{C}$.

We will focus on the evolution of a chromosome $\mathcal{C}^0$ of $\dot{\mathcal{C}}$ over time, when we suppose that  intrachromosomic inversions can occur. 
These \emph{inversions} have the form:

\begin{flushleft}
$(n_0, \hdots, n_{i-1}, \underline{n_i, n_{i+1}\hdots, n_{j-1}, n_j}, n_{j+1}, \hdots, n_\mathsf{N}) \longrightarrow$
\end{flushleft}
\begin{flushright}
$(n_0, \hdots, n_{i-1}, \underline{n_j, n_{j-1} \hdots, n_{i+1}, n_i,} n_{j+1}, \hdots, n_\mathsf{N}).$
\end{flushright}

The sequence of inversions corresponds to the sequence of segments $\left(\llbracket a^i ; b^i \rrbracket\right)_{i \in \mathds{N}}$, where $\llbracket a^i ; b^i \rrbracket$ represents the nucleotides segment that mutates at time $i$: the nucleotides from $n_{a^i}$ to $n_{b^i}$ are inverted.

Let $\mathcal{S}_{\mathsf{N}} = \left( \llbracket 1 ; \mathsf{N} \rrbracket \times \llbracket 1 ; \mathsf{N} \rrbracket \right)^\mathds{N}$ be the set of all the possible evolutions by inversion over time, $\mathcal{C}$ be a chromosome with $\mathsf{N}$ nucleotides, and $\mathcal{X}(\mathcal{C}) = \dot{\mathcal{C}}\times \mathcal{S}$.

We define the \emph{global inversion} by:

$$
\begin{array}{cccc}
i: & \dot{\mathcal{C}} & \longrightarrow & \dot{\mathcal{C}} \\
   & (n_1, \hdots, n_\mathsf{N}) & \longmapsto & (n_\mathsf{N}, \hdots, n_1).
   \end{array}
$$

In other words, for a chromosome $\mathcal{C}$, $i(\mathcal{C})$ is the chromosome in which the first nucleotide becomes the last one, etc.
Let us notice that, as the DNA strain is always read in the $5'\rightarrow 3'$ direction, then $i(C) \neq C$.

Let us now define the \emph{partial inversion} function, as follows:

$$
\begin{array}{cccc}
f: & \dot{\mathcal{C}}\times \llbracket 1;\mathsf{N}\rrbracket^2 & \longrightarrow & \dot{\mathcal{C}} \\
   & \left((n_1, \hdots, n_\mathsf{N}), (a,b)\right) & \longmapsto & (n_1', \hdots, n_\mathsf{N}'),
   \end{array}
$$
\noindent with:
$$
n_k' = \left\{ 
\begin{array}{ll}
n_k & \text{if } k \notin \llbracket a ; b \rrbracket,\\
\left(\sigma^{\mathsf{N}-a-b+1}\circ i\right) (n_1, \hdots, n_\mathsf{N})_k & \text{else,}
\end{array}
\right.
$$

\noindent where $\sigma$ is the nucleotide circular shift:

$$
\begin{array}{cccc}
\sigma: & \dot{\mathcal{C}} & \longrightarrow & \dot{\mathcal{C}} \\
   & (n_1, \hdots, n_\mathsf{N}) & \longmapsto & (n_2, \hdots, n_\mathsf{N}, n_1).
   \end{array}
$$

So $f(C,(a,b))$ is the chromosome corresponding to the inversion of the segment $\llbracket a ; b \rrbracket$ into the chromosome $\mathcal{C}$.
Furthermore, for each $(i,j) \in \llbracket 1, \mathsf{N} \rrbracket^2$, we define $f_{(i,j)}(C) = f\left(C,(i,j)\right)$.

\begin{remark}
$f$ can be rewritten as 

$$
\begin{array}{cccc}
f: & \dot{\mathcal{C}}\times \llbracket 1; \mathsf{N}\rrbracket^2 & \longrightarrow & \llbracket 1; \mathsf{N}\rrbracket^2 \\
   & \left((n_1, \hdots, n_\mathsf{N}),(a,b)\right) & \longmapsto & \left(n_k\left(1-\mathcal{I}_{\llbracket a;b\rrbracket}(k)\right)+\left(\sigma^{\mathsf{N}-a-b+1}\circ i\right)(n_1, \hdots, n_{\mathsf{N}})_k \mathcal{I}_{\llbracket a;b\rrbracket}(k)\right)_{k=1, \hdots, \mathsf{N}}.
   \end{array}
$$
where $\mathcal{I}_X$ is the indicator function of the set X:
$$
\mathcal{I}_X(x) = \left\{ 
\begin{array}{ll}
1 & \text{if } x \in X\\
0 & \text{else.}
\end{array}
\right.$$
\end{remark}

The \emph{inversion operator} $\mathfrak{I}$ is finally defined, for a given family of equivalent chromosomes $\dot{\mathcal{C}}$ having $\mathsf{N}$ nucleotides, by:

$$
\begin{array}{cccc}
\mathfrak{I}: & \mathcal{X}\left(\mathcal{C}\right) & \longrightarrow & \mathcal{X}\left(\mathcal{C}\right)\\
& \left( (n_1, \hdots, n_\mathsf{N}) ; \left(S^0, S^1, \hdots \right) \right) & \longmapsto &  \left( f\left((n_1, \hdots, n_\mathsf{N}),S^0\right) ; \left(S^1, S^2, \hdots \right) \right),
\end{array}
$$
that is, $\mathfrak{I}(C,S) = \left(f(C,S^0) ; \Sigma(S) \right)$, where $\Sigma \left( (S^n)_{n \in \mathds{N}}\right) = (S^{n+1})_{n \in \mathds{N}}$.

\subsection{A metric for chromosomes}

We can define a distance $d$ on $ \mathcal{X}\left(\mathcal{C}\right)$ by:
$\forall A=(C^A,S^A), B=(C^B,S^B) \in \dot{\mathcal{C}}\times \mathcal{S}_{\mathsf{N}},$
$$ d(A,B) = d_C(C^A,C^B) + d_S(S^A,S^B),$$

\noindent where:

\begin{itemize}
\item $d_C \left( (n_1, \hdots, n_\mathsf{N});(n_1', \hdots, n_\mathsf{N}') \right) = \displaystyle{\sum_{i=1}^\mathsf{N}} \delta(n_i,n_i')$, with $\delta(n,n') = 0$ if $n=n'$, and $\delta(n,n') = 1$ else.
\item $d_S(S,\check{S}) = \displaystyle{\dfrac{9}{\mathsf{N}}\sum_{i=1}^\infty \left( \dfrac{1}{10^i}\sum_{j=1}^\mathsf{N} \delta\left(g(S^i)_j; g(\check{S}^i)_j\right)\right)}$, with: $$g(a,b) = (1, \hdots, a-1, b,b-1, \hdots a+1,a,b+1, \hdots, \mathsf{N}).$$
\end{itemize}

\begin{proposition}
$d$ is a distance on $\mathcal{X}\left(\mathcal{C}\right)$.
\end{proposition}

\begin{proof}
We will show that $d$ is the sum of two distances.
\begin{enumerate}
\item Let us firstly demonstrate that $d_C$ is a distance:
\begin{itemize}
\item $d_C(C,C') = 0 \Rightarrow \forall i, \delta(n_i,n_i')=0$. So, $\forall i, n_i=n_i'$, and thus $C=C'$.
\item $d_C \left( C,C' \right) = \displaystyle{\sum_{i=1}^\mathsf{N}} \delta(n_i,n_i')= \displaystyle{\sum_{i=1}^\mathsf{N}} \delta(n_i',n_i)=d_C(C',C).$
\item $\forall C,C',C'' \in \mathcal{X}\left(\mathcal{C}\right), d_C(C,C'') \leqslant d_C(C,C')+d_C(C',C'')$. Indeed, $\forall i \in \llbracket 1; \mathsf{N} \rrbracket, \delta (n_i,n_i'') \leqslant \delta (n_i,n_i') + \delta (n_i',n_i'')$, because:
\begin{itemize}
\item it is obvious if $\delta(n_i,n_i'') = 0$,
\item else, $\delta(n_i,n_i'') = 1$, which implies either $n_i \neq n_i'$ or $n_i' \neq n_i''$. And so, either $\delta(n_i,n_i') = 1$ or $\delta(n_i',n_i'') = 1$.
\end{itemize}
\end{itemize}
\item Let us now prove that $d_S$ is a distance too. Obviously, $d_S(S,S') = d_S(S',S)$ and $d_(S,S) = 0.$ Finally, the triangle inequality of $d_S$ is inherited from the triangle inequality of $\delta$.
\end{enumerate}
\end{proof}

\begin{proposition}
The inversion operator $\mathfrak{I}$ is a continuous function on $\left(\mathcal{X}\left(\mathcal{C}\right),d\right)$.
\end{proposition}

\begin{proof}
Let $(C^k,S^k) \rightarrow (C,S)$. Then $d_C(C^k,C) \rightarrow 0$ and $d_S(S^k,S) \rightarrow 0$.
\begin{itemize}
\item One the one hand, as $d_C(C^k,C) \rightarrow 0$ and due to the fact that $d_C$ is an integer metric, we have: $\exists k_0, k \geqslant k_0 \Rightarrow C^k=C$. Additionally, as $d_S(S^k,S) \rightarrow 0$, $\exists k_1 \in \mathds{N}, k \geqslant k_1: d_S(S^k,S) < 10^{-1}.$ So $$\displaystyle{\dfrac{1}{10} \sum_{j=1}^\mathsf{N} \delta \left( g\left((S^k)^0\right)_j;g\left(S^0\right)_j\right) = 0.}$$ In other words, $\forall j \in \llbracket 1; \mathsf{N} \rrbracket, g\left((S^k)^0\right)_j = g\left(S^0\right)_j$. And thus $(S^k)^0 = S^0$. Finally, $\forall k \geqslant max(k_0,k_1), f\left(C^k, (S^k)^0\right) = f\left(C^k, S^0\right) \Rightarrow$ $$lim_{k\rightarrow \infty} f\left(C^k, (S^k)^0\right) = lim_{k\rightarrow \infty} f\left(C^k, S^0\right).$$
\item On the other hand, 
$$
\begin{array}{ll}
0 \leqslant d_S(\Sigma(S^k),\Sigma(S)) & = \displaystyle{\dfrac{9}{\mathsf{N}} \sum_{i=1}^\infty \left( \dfrac{1}{10^i}  \sum_{j=1}^\mathsf{N} \delta\left(g\left(\left(\Sigma S^k\right)^i\right)_j;g\left(\left(\Sigma S\right)^i\right)_j\right)\right)}\\
 & = \displaystyle{\dfrac{90}{\mathsf{N}} \sum_{i=1}^\infty \left( \dfrac{1}{10^{i+1}}  \sum_{j=1}^\mathsf{N} \delta\left(g\left(\left(\Sigma S^k\right)^{i+1}\right)_j;g\left(\left(\Sigma S\right)^{i+1}\right)_j\right)\right)}\\
 & = \displaystyle{\dfrac{90}{\mathsf{N}} \sum_{i=2}^\infty \left( \dfrac{1}{10^{i}}  \sum_{j=1}^\mathsf{N} \delta\left(g\left(\left(\Sigma S^k\right)^{i}\right)_j;g\left(\left(\Sigma S\right)^{i}\right)_j\right)\right)}\\
 & \leqslant \displaystyle{\dfrac{90}{\mathsf{N}} \sum_{i=1}^\infty \left( \dfrac{1}{10^{i}}  \sum_{j=1}^\mathsf{N} \delta\left(g\left(\left(\Sigma S^k\right)^{i}\right)_j;g\left(\left(\Sigma S\right)^{i}\right)_j\right)\right)} \rightarrow 0 
\end{array}
$$
\end{itemize}
We can thus conclude to the continuity of $\mathfrak{I}$ on $\left(\mathcal{X}\left(\mathcal{C}\right),d\right)$.
\end{proof}

Let us now introduce two lemmas. Their proofs are obvious.

\begin{lemma}
Any transposition $(i,j)$ can be written as a composition of $f_{\llbracket i+1;j-1\rrbracket} \circ f_{\llbracket i;j\rrbracket}$.
\end{lemma}

\begin{lemma}
Any permutation can be written as a composition of $f_{\llbracket i;j\rrbracket}$.
\end{lemma}

\section{Chaos of DNA inversion}

\begin{proposition}
The inversion operator $\mathfrak{I}$ is strongly transitive on $\left(\mathcal{X}\left(\mathcal{C}\right),d\right)$.
\end{proposition}

\begin{proof}
Let $A=(N^A,S^A)$ and $B=(N^B,S^B)$ two points of $\mathcal{X}\left(\mathcal{C}\right)$, and $\varepsilon >0$. We define $N = N^A$, and $\forall k \leqslant k_0 = - \left[ log_{10}(\varepsilon)\right], S^k = (S^A)^k$. Let $N' = \mathfrak{I}^{k_0}(N^A,S^A)_1$. There is a permutation that maps $N'$ on $N^B$, so there exists $S'=(S_1',S_2', \hdots, S_{k_1}') \in \llbracket 1, \mathsf{N} \rrbracket^{k_1}$ such that $\mathfrak{I}^{k_1}(N',S')_1 = N^B$.

Then the point:
\begin{itemize}
\item $N = N^A$,
\item $\forall k \leqslant k_0, S^k = (S^A)^k$,
\item $\forall k \in \llbracket 1, k_1 \rrbracket, S^{k_0+k} = S'^k$,
\item $\forall k \in \mathds{N}, S^{k_0+k_1+k+1} = (S^B)^k$.
\end{itemize}
is $\varepsilon-$close to $A$, and such that $\mathfrak{I}^{k_0+k_1}(N,S) = B.$
\end{proof}

\begin{proposition}
The inversion operator $\mathfrak{I}$ is regular on $\left(\mathcal{X}\left(\mathcal{C}\right),d\right)$.
\end{proposition}

\begin{proof}
Let $A=(N^A,S^A) \in \mathcal{X}\left(\mathcal{C}\right)$ and $\varepsilon > 0$. We define $ k_0 = - \left[ log_{10}(\varepsilon)\right]$ and $\tilde{N} = \mathfrak{I}^{k_0}(A)_1$. A permutation can be found that maps $\tilde{N}$ into $N^A$, then there exists $\tilde{S}=(\tilde{S}^1, \hdots, \tilde{S}^{k_1})$ such that $\mathfrak{I}^{k_1}(N,\tilde{S})_1 = A$.

Then the point $(N,S)$ defined by:
\begin{itemize}
\item $N = N^A$,
\item $\forall k \leqslant k_0, S^k = (S^A)^k,$
\item $\forall k \in \llbracket 1; k_1 \rrbracket, S^{k_0+k} = \tilde{S}^k$,
\item $\forall k \in \mathds{N}, S^{k_0+k_1+k+1} = S^k$,
\end{itemize}
is a periodic point $\varepsilon-$close to $A$.
\end{proof}

As the inversion dynamic is both transitive and regular, we can thus conclude that,

\begin{theorem}
DNA inversion is chaotic, as it is defined in the Devaney's theory.
\end{theorem}

\section{Consequences}

\section{Conclusion}

\bibliographystyle{plain}
\bibliography{mabase}

\end{document}